\documentclass[lettersize,journal]{IEEEtran}
\usepackage{amsmath,mathtools,amsfonts,amssymb}
\usepackage{amsthm}

\usepackage{array}
\usepackage{xcolor}
\usepackage[caption=false,font=normalsize,labelfont=sf,textfont=sf]{subfig}
\usepackage{textcomp}
\usepackage{stfloats}
\usepackage{url}
\usepackage{verbatim}
\usepackage{graphicx}
\usepackage{cuted}
\usepackage{float} 
\usepackage{bbm}
\usepackage{multirow}
\usepackage{threeparttable} 

\def\BibTeX{{\rm B\kern-.05em{\sc i\kern-.025em b}\kern-.08em
    T\kern-.1667em\lower.7ex\hbox{E}\kern-.125emX}}
\usepackage{balance}
\usepackage{booktabs} 
\usepackage{caption}                         
\usepackage{algpseudocode}
\newtheoremstyle{sltheorem}
{}                
{}                
{}        
{10pt}                
{\bfseries}       
{:}               
{ }               
{}                
\theoremstyle{sltheorem}

\newtheorem{definition}{Definition}

\newtheorem{assumption}{Assumption}

\newtheorem{proposition}{Proposition}
\AtBeginDocument{%
  \setlength{\abovedisplayskip}{3pt plus 2pt minus 1pt}%
  \setlength{\belowdisplayskip}{3pt plus 2pt minus 1pt}%
  \setlength{\abovedisplayshortskip}{0pt plus 1pt}%
  \setlength{\belowdisplayshortskip}{2pt plus 1pt minus 1pt}%
}
\begin{document}
\title{A Gradient-Based Capacity Accreditation Framework in Resource Adequacy: Formulation, Computation, and Practical Implications}
\author{Qian Zhang, Feng Zhao, Gord Stephen, Chanan Singh, Le Xie
\thanks{Qian Zhang and Le Xie are with the School of Engineering and Applied Sciences, Harvard University, Allston, USA. Feng Zhao is with the ISO New England Inc., Holyoke, MA, USA. Gord Stephen is with the Department of Electrical and Electronic Engineering, Imperial College London, SW7 2BU London, UK. Chanan Singh is with the Department of Electrical and Computer Engineering, Texas A\&M University, College Station, USA. (correspondence e-mail: qianzhang@g.harvard.edu).}}

\markboth{Working Paper, Jan 2026}%
{How to Use the IEEEtran \LaTeX \ Templates}

\maketitle

\markboth{Working Paper, Jan 2026}%
{How to Use the IEEEtran \LaTeX \ Templates}

\begin{abstract}
Probabilistic resource adequacy assessment is a cornerstone of modern capacity accreditation. This paper develops a gradient-based framework, in which capacity accreditation is interpreted as the directional derivative of a probabilistic resource adequacy metric with respect to resource capacity, that unifies two widely used accreditation approaches: Effective Load Carrying Capability (ELCC) and Marginal Reliability Impact (MRI). Under mild regularity conditions, we show that marginal ELCC and MRI yield equivalent accreditation factors, while their numerical implementations exhibit markedly different computational characteristics. Building on this framework, we demonstrate how infinitesimal perturbation analysis enables up to a $1000\times$ speedup in gradient estimation for capacity accreditation, and we implement gradient-informed search algorithms that significantly accelerate ELCC computations relative to standard bisection methods. Large-scale Monte Carlo experiments show that MRI achieves substantial runtime reductions compared to ELCC and exhibits greater robustness to perturbation step-size selection. These results provide practical guidance for implementing efficient and scalable capacity accreditation in large-scale power systems.
\end{abstract}


\section{Introduction}
To ensure power system reliability and cost-efficiency under a rapidly changing resource mix, many regions have been reforming \emph{capacity accreditation} frameworks to improve fairness and accuracy, both in market-based settings \cite{pjm_market_design_project_roadmap, iso_ne_capacity_auction_reforms, zuo2025revisiting} and in centralized planning frameworks \cite{nelson2021rawg,zhang2023power}. A key motivation is that legacy heuristic approaches, such as outage–rate–discounted unforced capacity (UCAP) for thermal units and historical average output for intermittent resources, can mischaracterize reliability contributions as systems become increasingly weather-driven and exhibit stronger cross-resource correlations. In principle, accredited capacities from different resource types are treated as indistinguishable for meeting system capacity requirements, implying that one megawatt of capacity credit should provide the same contribution to system resource adequacy (RA) independent of the underlying technology.

To address these challenges, two \emph{gradient-based} probabilistic capacity accreditation methods are now widely used (or under active consideration): \emph{Effective Load Carrying Capability} (ELCC) and \emph{Marginal Reliability Impact} (MRI). Both quantify a resource’s contribution to RA by evaluating changes in a chosen reliability metric under a probabilistic loss-of-load simulation model. Moreover, when the reliability metric is selected as Expected Unserved Energy (EUE), marginal ELCC yields the same accreditation value as MRI under standard regularity conditions \cite{fengmri}. Despite this conceptual relationship, the literature has not yet established a rigorous mathematical formulation, nor provided a clear comparison of the computational complexity between ELCC and MRI in large-scale system applications.

A resource’s ELCC is defined as the greatest constant firm load that can be added to the system \emph{after} adding the resource while maintaining the same reliability level \cite{garver2007effective}. Capacity accreditation based on ELCC is obtained by normalizing the load-carrying capability by the nameplate quantity of the resource addition, whether evaluated for an entire resource class (average ELCC) or for an incremental unit (marginal ELCC)\footnote{Unless otherwise stated, ``ELCC'' in this paper refers to \emph{marginal} ELCC.}.
In contrast, MRI is defined as the marginal impact of a small capacity variation on system EUE \cite{zhao2017constructing}. Unlike ELCC (measured in MW), MRI has units of hours per study period, and the MRI-based accreditation is defined as the normalized MRI by the MRI of a \emph{perfect} reference resource \cite{fengmri}.

ELCC and MRI are evaluated through probabilistic assessment of stochastic power systems with heterogeneous resources and discrete outage and availability events, implemented in practice using probabilistic resource adequacy assessment (RAA) simulators such as GE-MARS \cite{GEMARS} and PRAS \cite{stephen2024probabilistic}. A substantial body of RA research has focused on the definition and evaluation of reliability metrics \cite{Gord9810615, nerc2024evolving}, Monte Carlo sample size selection \cite{stephen2024sample}, and the integration of emerging resources such as variable generation and storage \cite{singh1985reliability,xu2013power}. These efforts primarily target accurate estimation of system-level reliability outcomes under uncertainty.

Capacity accreditation adds a further difficulty: it must quantify how reliability \emph{changes} under marginal additions, a \emph{gradient estimation} problem in stochastic simulation. The simulation literature offers several estimator families, finite differences, infinitesimal perturbation analysis, and likelihood-ratio methods, with distinct bias--variance and implementability tradeoffs \cite{fu2012conditional,mohamed2020monte}, but rigorous formulations connecting them to capacity accreditation, with large-scale validation under realistic RA settings, remain limited \cite{salehfar2002application}.

The contributions are fourfold. First, we formulate ELCC- and MRI-based accreditation within a unified gradient-based framework by casting accreditation as directional derivatives on a well-defined probability space, with explicit assumptions under which marginal ELCC and MRI are equivalent (Sections~\ref{sec:preliminaries} and \ref{sec:CA}). Second, we introduce an infinitesimal perturbation analysis (IPA) perspective that clarifies the meaning of these gradients and when unbiased pathwise estimators are available, enabling accreditation from a single simulation run (Section~\ref{sec:IPA}). Third, for the perturbation-based implementations that dominate practice, we give a comparative complexity analysis showing why ELCC requires iterative root-finding while MRI is single-pass (Section~\ref{sec:CC}). Finally, large-scale studies quantify computation time and numerical sensitivity for ELCC and MRI across perturbation step sizes and baseline reliability levels, with practical implications for scalable accreditation (Section~\ref{sec:nc}).

\section{Preliminaries and Definitions} \label{sec:preliminaries}
\subsection{System Descriptions}
Let $\mathcal{T}$ denote the set of all hours in a year (or the study periods of interest), $\mathcal{D}$ the set of all days, and $\Omega$ the set of all possible year-long scenarios for the system under study, assumed to be finite. Following the notation in \cite{zuo2025revisiting}, we construct a probability space $(\Omega, \mathcal{F}, \mathbb{P})$, where $\mathcal{F}$ represents the $\sigma$-algebra of events on $\Omega$ and $\mathbb{P}$ is a probability measure defined on $\mathcal{F}$ satisfying standard measure-theoretic properties. Let $X: \Omega \rightarrow \mathbb{R}$ be a random variable representing our RA metric for any given scenario. For each generator $g \in \mathcal{G}$, let $x_g \in \mathbb{R}_{+}$ denote its installed capacity. In the absence of storage, and assuming that generation resources have no intertemporal operating constraints, we define
\[
A_{g\tau\omega} : \mathcal{T} \times \Omega \rightarrow [0,1]
\]
as a stochastic process representing generator $g$'s availability, where $A_{g\tau\omega}$ denotes its realization in hour $\tau$ of scenario $\omega$. Then, the total available capacity is stochastic processes on $\mathcal{T} \times \Omega \rightarrow \mathbb{R}_{+}$, defined as
\[
\hat{x}_{\tau\omega} = \sum_{g \in \mathcal{G}} x_g A_{g\tau\omega}.
\]

The system load profile is assumed to be sampled from a set of representative scenarios, with $L_{\tau\omega}$ denoting load in hour $\tau$ of scenario $\omega$. For notational compactness, we define the vectors of hourly values within scenario $\omega$ as 
\[
\boldsymbol{\hat{x}}_\omega := (\hat{x}_{\tau\omega})_{\tau \in \mathcal{T}},
\quad 
\boldsymbol{L}_\omega := (L_{\tau\omega})_{\tau \in \mathcal{T}},
\]
representing total available capacity and system load, respectively.

\subsection{Resource Adequacy Metrics}

For semantic clarity, we distinguish three related but distinct concepts: the \emph{RA metric}, the \emph{risk measure}, and the \emph{reliability standard}. The RA metric specifies how shortage events are quantified and aggregated within a single assessment period, such as Unserved Energy (UE), Loss of Load Hours (LOLH), and Loss of Load Days (LOLD). The risk measure defines how these per-scenario outcomes are aggregated across multiple stochastic realizations (e.g., expectation or percentile). Finally, the reliability standard specifies the target or benchmark value of the risk measure that the system is expected to satisfy, such as the `one day in ten years' LOLD criterion commonly used in North America through the
1960s \cite{Gord9810615}.

We first define the RA metric for each scenario $\omega$ as a functional
\[
M: \mathbb{R}_+^{|\mathcal{T}|} \times \mathbb{R}_+^{|\mathcal{T}|} \to \mathbb{R}_+,
\quad
M(\boldsymbol{\hat{x}}_\omega, \boldsymbol{L}_\omega),
\]
which aggregates unserved load over time according to the chosen reliability measure. Examples include:
\begin{align*}
M_{\text{LOLH}}(\boldsymbol{\hat{x}}_\omega, \boldsymbol{L}_\omega) 
&= \sum_{\tau \in \mathcal{T}} \mathbbm{1}\Big\{ (L_{\tau\omega} - \hat{x}_{\tau\omega}) > 0 \Big\}, \\
M_{\text{LOLD}}(\boldsymbol{\hat{x}}_\omega, \boldsymbol{L}_\omega) 
&= \sum_{d \in \mathcal{D}} \mathbbm{1}\Big\{ \max_{\tau \in \mathcal{T}_d} (L_{\tau\omega} - \hat{x}_{\tau\omega})> 0 \Big\}, \\
M_{\text{UE}}(\boldsymbol{\hat{x}}_\omega, \boldsymbol{L}_\omega) 
&= \sum_{\tau \in \mathcal{T}} (L_{\tau\omega} - \hat{x}_{\tau\omega})_+,
\end{align*}
where $(\cdot)_+ := \max\{\cdot,0\}$.

The system-level RA measure is defined by aggregating the per-scenario metric across the probability space:
\[
\mathcal{M}(\boldsymbol{\hat{x}}, \boldsymbol{L})
:= \Phi_\omega \big[\, M(\boldsymbol{\hat{x}}_\omega, \boldsymbol{L}_\omega) \,\big],
\]
where $\Phi_\omega[\cdot]$ is a risk operator. Common choices include the expectation
$\Phi_\omega[Y]=\mathbb{E}_\omega[Y]$, or a tail-risk measure such as the conditional value-at-risk (CVaR).

\section{Gradient-Based Capacity Accreditation Methods} \label{sec:CA}

\subsection{Assumptions}

Capacity accreditation is performed relative to a baseline system configuration. In practice, the system operator first estimates the baseline system’s Installed Capacity Requirement (ICR) prior to accreditation using historical information, planned retirements, and policy-driven assumptions (e.g., renewable portfolio targets). Throughout this paper, we let $\boldsymbol{\hat{x}}$ denote the operator’s baseline estimate of available capacity trajectories used in the accreditation study. This baseline plays a central role in marginal-based accreditation methods, as both ELCC and MRI quantify reliability impacts relative to this reference system. We adopt the notation of Section~\ref{sec:preliminaries} and impose the following assumptions.

\begin{assumption}[Baseline Shortfall and Shortfall Dependence]\label{ass:shortfall}
Assume the baseline system is not perfectly adequate, i.e., $\mathcal{M}(\boldsymbol{\hat{x}},\boldsymbol{L})>0$, and that there exists a functional $\widetilde{\mathcal{M}}$ such that, for any pair of trajectories $(\boldsymbol{\hat{x}},\boldsymbol{L})$,
\begin{equation}
\mathcal{M}(\boldsymbol{\hat{x}},\boldsymbol{L})
=\widetilde{\mathcal{M}}\big(\boldsymbol{L}-\boldsymbol{\hat{x}}\big).
\end{equation}
Equivalently, for any scalar $c\in\mathbb{R}$,
\begin{equation}
\mathcal{M}(\boldsymbol{\hat{x}}+c\mathbf{1},\,\boldsymbol{L}+c\mathbf{1})
=\mathcal{M}(\boldsymbol{\hat{x}},\,\boldsymbol{L}).
\end{equation}
\end{assumption}

This assumption holds for standard RA metrics combined with common risk operators (e.g., expectation or quantiles), since these measures depend only on shortfalls. {The shortfall-dependence structure $\mathcal{M}=\widetilde{\mathcal{M}}(\boldsymbol{L}-\boldsymbol{\hat{x}})$ holds exactly when per-period shortfall is set by the contemporaneous capacity--load balance, i.e., absent binding intertemporal constraints (ramping and minimum up/down times) \cite{zhang2026comparative}. When such constraints bind, the gradient framework nonetheless applies: the directional derivative in Assumption~\ref{ass:diff} remains well defined whenever the simulator-embedded operating policy responds continuously to small perturbations, with the perturbation interpreted as a \emph{dispatch-induced} direction. We retain Assumption~\ref{ass:shortfall} to establish a unified treatment of the ELCC–EFC equivalence.}

\begin{assumption}[Monotonicity and Continuity]\label{ass:mono}
$\mathcal{M}$ is nonincreasing in $\boldsymbol{\hat{x}}$ and nondecreasing in $\boldsymbol{L}$ (componentwise), and is continuous in both arguments.
\end{assumption}

Assumption~\ref{ass:mono} also holds for non-storage resources under standard RA metrics with common risk operators. For storage resources, however, monotonicity must be verified with respect to the assumed dispatch policy, as the effective contribution of storage depends on intertemporal operating decisions \cite{qianstorage}.

\begin{assumption}[Directional Differentiability at the Baseline]\label{ass:diff}
At the baseline $(\boldsymbol{\hat{x}},\boldsymbol{L})$, the risk measure $\mathcal{M}$ admits (Gâteaux) directional derivatives in the directions relevant for capacity accreditation. Specifically, for any direction $\boldsymbol{v}\in\mathbb{R}_+^{|\mathcal{T}|}$ of interest,
\begin{equation}\label{eq:dir-deriv}
\partial_{\hat{x}}\mathcal{M}[\boldsymbol{v}]
:=\left.\frac{d}{d\epsilon}\right|_{\epsilon=0}\mathcal{M}(\boldsymbol{\hat{x}}+\epsilon\boldsymbol{v},\boldsymbol{L})
\end{equation}
exists.
Here, $\boldsymbol{v}$ represents a capacity-availability perturbation associated with a marginal change in a single resource, a resource class, or a colocated portfolio (e.g., $v_\tau=A_{g\tau}$ for a candidate unit $g$).
We further assume that the baseline evaluation point is \emph{regular} in the sense that it does not coincide with a kink of the underlying piecewise-linear mapping; in particular, this is typically satisfied when the baseline system exhibits nonzero shortfall under the chosen metric.
\end{assumption}

In this paper, we primarily focus on EUE, for which $\mathcal{M}$ is piecewise linear in $\boldsymbol{\hat{x}}$ and may fail to be differentiable \emph{only} at turning points where marginal capacity does change the set of shortage hours \cite{zachary2022integration}. Since accreditation evaluates derivatives at a fixed baseline and along specific directions, assuming regularity at the baseline is mild in practice. {For counting-based metrics, such as LOLD, although the \emph{per-scenario} metric $M$ is a step function of capacity (kinked wherever a marginal MW changes the shortage-hour set), the \emph{system-level} $\mathcal{M}=\Phi_\omega[M]$ averages these jumps: provided the hourly margins $L_{\tau\omega}-\hat{x}_{\tau\omega}$ have no atom at zero, $\mathcal{M}$ is differentiable almost everywhere, with kinks confined to a measure-zero set. In finite samples this smoothing is approximate, so counting-metric gradients are noisier than EUE and need larger steps or sample sizes, a numerical rather than structural limitation that we quantify in Section~\ref{sec:lole-results}.}

\subsection{Effective Load Carrying Capability (ELCC)}

The ELCC of a candidate unit quantifies the largest constant load $L_c$ that can be added while maintaining the same system-level reliability after adding the unit.

\begin{definition}[Marginal ELCC]\label{def:elcc}
Let the candidate unit have additional nameplate capacity $\Delta x>0$ and scenario-wise availability $\boldsymbol{A}_{g,\omega}=(A_{g\tau\omega})_{\tau\in\mathcal{T}}$. The \emph{largest} $L_c$ meeting
\begin{equation}\label{eq:elcc-implicit}
\mathcal{M}\big(\boldsymbol{\hat{x}}+\Delta x\,\boldsymbol{A}_{g},\,\boldsymbol{L}+L_c\,\mathbf{1}\big)
=
\mathcal{M}\big(\boldsymbol{\hat{x}},\,\boldsymbol{L}\big)
\end{equation}
is defined as the value ELCC. Then, the associated \emph{ELCC Capacity Accreditation} is the fraction of nameplate credited:
\begin{equation}
\alpha_g^{\mathrm{ELCC}}:=\frac{L_c}{\Delta x}\in[0,1].
\end{equation}
\end{definition}

\begin{proposition}[Equivalent ELCC formulation]\label{prop:elcc-equivalent}
Under Assumption~\ref{ass:shortfall}, \eqref{eq:elcc-implicit} is equivalent to
\begin{equation}\label{eq:elcc-shifted}
\mathcal{M}\big(\boldsymbol{\hat{x}}+\Delta x\,\boldsymbol{A}_{g},\,\boldsymbol{L}\big)
=
\mathcal{M}\big(\boldsymbol{\hat{x}}+L_c\mathbf{1},\,\boldsymbol{L}\big),
\end{equation}
where $\boldsymbol{\hat{x}}+L_c\mathbf{1}$ corresponds to adding $L_c$ MW of a resource uniformly available across all periods and without outage. 
\end{proposition} 

\emph{Remark}: The value $L_c$ defined by \eqref{eq:elcc-shifted} is often referred to as the marginal \emph{Equivalent Firm Capacity} (EFC). Under Assumption~\ref{ass:shortfall}, marginal EFC is numerically equivalent to marginal ELCC \emph{locally}, although the two concepts differ in physical interpretation: ELCC is framed as an equivalent load increase, whereas EFC is framed as an equivalent addition of perfectly firm capacity \cite{zachary2012probability}.

\begin{proposition}[Existence and Uniqueness of $L_c$ for EUE]\label{prop:exist-unique}
Fix $\Delta x>0$ and let $\mathcal{M}$ be the EUE-based risk measure. Under Assumptions~\ref{ass:shortfall}, \ref{ass:mono}, and \ref{ass:diff}, there exists a unique $L_c\ge 0$ satisfying \eqref{eq:elcc-implicit} or \eqref{eq:elcc-shifted}.
\end{proposition}

\begin{proof}
Consider the function
\[
\phi(c)=\ \mathcal{M}(\boldsymbol{\hat{x}}+c\mathbf{1},\,\boldsymbol{L})
-\mathcal{M}(\boldsymbol{\hat{x}}+\Delta x\,\boldsymbol{A}_{g},\,\boldsymbol{L}),
\quad c\ge 0.
\]
By Assumption~\ref{ass:mono}, $\phi$ is continuous and nonincreasing in $c$ (adding firm capacity cannot increase EUE). Moreover, $\phi(0)\ge 0$ since $\Delta x\,\boldsymbol{A}_g\ge \boldsymbol{0}$ implies
$\mathcal{M}(\boldsymbol{\hat{x}},\boldsymbol{L})\ge \mathcal{M}(\boldsymbol{\hat{x}}+\Delta x\,\boldsymbol{A}_{g},\boldsymbol{L})$.
As $c\to+\infty$, EUE vanishes, so $\mathcal{M}(\boldsymbol{\hat{x}}+c\mathbf{1},\boldsymbol{L})\to 0$ and hence
$\phi(c)\to -\,\mathcal{M}(\boldsymbol{\hat{x}}+\Delta x\,\boldsymbol{A}_{g},\boldsymbol{L})<0$,
where strict negativity uses Assumption~\ref{ass:shortfall} and $\Delta x<\infty$.
Therefore, by the intermediate value theorem, there exists $L_c\ge 0$ such that $\phi(L_c)=0$,
which is exactly \eqref{eq:elcc-shifted}.

For uniqueness, Assumption~\ref{ass:diff} implies that $\phi$ has directional derivative
\[
\phi'(c)=\partial_{\hat{x}}\mathcal{M}[\mathbf{1}]\Big|_{(\boldsymbol{\hat{x}}+c\mathbf{1},\,\boldsymbol{L})}.
\]
For the EUE metric, $\partial_{\hat{x}}\mathcal{M}[\mathbf{1}]<0$ whenever there is positive probability of shortfall,
so $\phi$ is strictly decreasing on the relevant range. Hence $\phi(c)=0$ has at most one solution, and $L_c$ is unique.
\end{proof}

\begin{proposition}[First-order ELCC response]\label{prop:elcc-linear}
Assume $\mathcal{M}$ is the EUE-based risk measure. Under Assumptions~\ref{ass:shortfall} and \ref{ass:diff}, as $\Delta x\to 0$ the ELCC $L_c$ satisfying \eqref{eq:elcc-shifted} obeys
\begin{equation}
L_c
=\frac{\partial_{\hat{x}}\mathcal{M}[\boldsymbol{A}_{g}]}{\partial_{\hat{x}}\mathcal{M}[\mathbf{1}]}\,\Delta x
\;+\; o(\Delta x),
\end{equation}
and therefore
\begin{equation}
\alpha_g^{\mathrm{ELCC}}
=\frac{L_c}{\Delta x}
=\frac{\partial_{\hat{x}}\mathcal{M}[\boldsymbol{A}_{g}]}{\partial_{\hat{x}}\mathcal{M}[\mathbf{1}]}
\;+\; o(1).
\end{equation}
\end{proposition}

\begin{proof}
Using the equivalent ELCC formulation \eqref{eq:elcc-shifted}, define
\[
F(\Delta x,c)\ :=\ \mathcal{M}(\boldsymbol{\hat{x}}+\Delta x\,\boldsymbol{A}_{g},\boldsymbol{L})
-\mathcal{M}(\boldsymbol{\hat{x}}+c\mathbf{1},\boldsymbol{L}).
\]
Then $F(0,0)=0$ and the ELCC condition is $F(\Delta x,L_c)=0$. By Assumption~\ref{ass:diff},
the directional derivatives at $(0,0)$ satisfy
\[
\left.\frac{\partial F}{\partial \Delta x}\right|_{(0,0)}=\partial_{\hat{x}}\mathcal{M}[\boldsymbol{A}_{g}],
\qquad
\left.\frac{\partial F}{\partial c}\right|_{(0,0)}=-\,\partial_{\hat{x}}\mathcal{M}[\mathbf{1}].
\]
A first-order expansion around $(0,0)$ gives
\[
0=F(\Delta x,L_c)
=\partial_{\hat{x}}\mathcal{M}[\boldsymbol{A}_{g}]\,\Delta x
-\partial_{\hat{x}}\mathcal{M}[\mathbf{1}]\,L_c
+o(\Delta x + L_c),
\]
which rearranges to
\[
L_c=\frac{\partial_{\hat{x}}\mathcal{M}[\boldsymbol{A}_{g}]}{\partial_{\hat{x}}\mathcal{M}[\mathbf{1}]}\,\Delta x
+o(\Delta x).
\]
Dividing by $\Delta x$ yields the expression for $\alpha_g^{\mathrm{ELCC}}$.
\end{proof}

\subsection{Marginal Reliability Impact (MRI)}

The MRI formalizes the directional sensitivity of the risk measure to an incremental addition of resource $g$.

\begin{definition}[MRI]\label{def:mri}
For $\Delta x>0$, define the finite-difference MRI of $g$ at $(\boldsymbol{\hat{x}},\boldsymbol{L})$ by
\begin{equation}
\mathrm{MRI}_g(\Delta x)
:=\frac{\mathcal{M}(\boldsymbol{\hat{x}}+\Delta x\,\boldsymbol{A}_{g},\boldsymbol{L})
-\mathcal{M}(\boldsymbol{\hat{x}},\boldsymbol{L})}{\Delta x}.
\end{equation}
Based on Assumption~\ref{ass:diff}, the MRI can be formulated as: 
\begin{equation}
\mathrm{MRI}_g
:= \partial_{\hat{x}}\mathcal{M}[\boldsymbol{A}_{g}]
= \lim_{\Delta x\rightarrow 0}\mathrm{MRI}_g(\Delta x).
\end{equation}
\end{definition}

\begin{definition}[Perfect Resource]\label{def:perfect}
A perfect resource delivers one unit in every hour, i.e. its availability trajectory is $\mathbf{1}$. Its marginal MRI is $ \mathrm{MRI}_{\mathrm{perf}}:=\partial_{\hat{x}}\mathcal{M}[\mathbf{1}]$.
\end{definition}

\begin{definition}[MRI-based Capacity Accreditation]\label{def:mri-accr}
The MRI-based accreditation factor of $g$ is defined by normalizing its marginal impact by the perfect resource:
\begin{equation}
\alpha_g^{\mathrm{MRI}}
:= \frac{\mathrm{MRI}_g}{\mathrm{MRI}_{\mathrm{perf}}}
= \frac{\partial_{\hat{x}}\mathcal{M}[\boldsymbol{A}_{g}]}{\partial_{\hat{x}}\mathcal{M}[\mathbf{1}]}.
\end{equation}
\end{definition}

\subsection{Equivalence of Capacity Accreditation}

We now state and prove the local equivalence between ELCC- and MRI-based capacity accreditation, a relationship that has also been discussed informally in \cite{ferc2025mri}.

\begin{proposition}[Equivalence of ELCC and MRI Accreditation]\label{prop:equivalence}
Assume $\mathcal{M}$ is the EUE-based risk measure. Under Assumptions~\ref{ass:shortfall}, \ref{ass:mono}, and \ref{ass:diff}, the marginal ELCC accreditation factor equals the MRI-based accreditation factor:
\[
\alpha_g^{\mathrm{ELCC}}
=\alpha_g^{\mathrm{MRI}}
=\frac{\partial_{\hat{x}}\mathcal{M}[\boldsymbol{A}_{g}]}{\partial_{\hat{x}}\mathcal{M}[\mathbf{1}]}.
\]
\end{proposition}

\begin{proof}
Taking $\Delta x\to 0$ and applying Proposition~\ref{prop:elcc-linear} and Definition~\ref{def:mri-accr} yields the claim.
\end{proof}

{\emph{Remark (Local Validity and Finite Additions)}:
Proposition~\ref{prop:equivalence} is local, characterizing an infinitesimal change at a fixed baseline. Since most practical capacity accreditation exercises rely on finite rather than infinitesimal perturbations, the resulting accreditation values may exhibit some dependence on the chosen perturbation size \cite{fengmri}. For non-marginal additions, $\mathcal{M}$ is nonlinear, so the baseline marginal factor generally differs from the average factor over the addition by a term determined by the curvature of $\mathcal{M}$. In particular, as added capacity reduces the scarcity hours that confer marginal value, marginal factors tend to decline along the addition path, and average ELCC can exceed the post-addition marginal value \cite{schlag2020capacity}. This affects where the gradient should be evaluated rather than undermining the gradient-based characterization: in market applications, the relevant baseline is the anticipated future system, including expected entry, at which the marginal factor prices the next increment at equilibrium. The impact of a large discrete addition can be recovered by integrating MRI along the addition path, or equivalently by evaluating gradients at selected intermediate points. The practical extent of this nonlinearity is examined through the baseline-load and resource-mix sensitivity studies in Section~\ref{sec:nc}.}

\section{Infinitesimal Perturbation Analysis} \label{sec:IPA}
Gradient-based capacity accreditation relies on the existence of a well-defined \emph{infinitesimal} sensitivity of the risk measure $\mathcal{M}(\boldsymbol{\hat{x}},\boldsymbol{L})$ with respect to available capacity. This section presents the infinitesimal perturbation analysis framework for characterizing and interpreting such sensitivities through unbiased pathwise derivatives.

\subsection{Infinitesimal Perturbation Analysis}
In a stochastic simulation setting, the key theoretical question is whether this gradient can be interpreted as the expectation of a scenario-wise (sample-path) derivative—i.e., whether differentiation can be interchanged with the risk operator (here, expectation). This interchange is the foundation of \emph{infinitesimal perturbation analysis} (IPA) and provides a rigorous meaning for ``marginal reliability impact'' as a true derivative rather than a finite-difference artifact \cite{fu2012conditional,mohamed2020monte}.

\begin{definition}[Infinitesimal Perturbation Analysis (IPA) \cite{fu2012conditional,mohamed2020monte}]\label{def:ipa}
Let $F(\boldsymbol{\hat{x}},\boldsymbol{L};\omega)$ denote a scenario-wise performance measure (e.g., unserved energy), and let the risk measure be
$\mathcal{M}(\boldsymbol{\hat{x}},\boldsymbol{L})=\mathbb{E}_\omega[F(\boldsymbol{\hat{x}},\boldsymbol{L};\omega)]$.
IPA refers to estimating a directional derivative of $\mathcal{M}$ using the corresponding \emph{pathwise} directional derivative of $F$, i.e.,
\[
\partial_{\hat{x}}\mathcal{M}[\boldsymbol{v}]
=\!\partial_{\hat{x}} \left[\mathbb{E}_\omega F[\boldsymbol{v}]\right]=\mathbb{E}_\omega\!\left[\partial_{\hat{x}}F[\boldsymbol{v}]\right],
\]
whenever the interchange of differentiation and expectation is valid. In this case, $\partial_{\hat{x}}F[\boldsymbol{v}]$ provides an unbiased single-scenario contribution to the gradient, enabling step-size-free sensitivity estimation.
\end{definition}

\begin{proposition}[IPA for EUE]
\label{prop:ipa}
Consider the EUE setting where the scenario-wise unserved energy is
\begin{equation}
\begin{aligned}
\mathcal{M}(\boldsymbol{\hat{x}},\boldsymbol{L})
:=\mathbb{E}_\omega\!\left[F(\boldsymbol{\hat{x}},\boldsymbol{L};\omega)\right]= \mathbb{E}_\omega\sum_{\tau\in\mathcal{T}}\big(L_{\tau\omega}-\hat{x}_{\tau\omega}\big)_+
\end{aligned}    
\end{equation}
Under Assumptions~\ref{ass:shortfall}, \ref{ass:mono}, and \ref{ass:diff}, the following mild regularity conditions hold: (i) $\mathbb{P}\{L_{\tau\omega}=\hat{x}_{\tau\omega}\}=0$ for all $\tau\in\mathcal{T}$, and (ii) $F(\boldsymbol{\hat{x}},\boldsymbol{L};\omega)$ is integrable and uniformly Lipschitz in $\boldsymbol{\hat{x}}$ in a neighborhood of the baseline. Consequently, differentiation and expectation can be interchanged. In particular, $\mathcal{M}$ admits a directional derivative in any direction $\boldsymbol{v}\in\mathbb{R}^{|\mathcal{T}|}$ and
\begin{equation}\label{eq:ipa-eue}
\partial_{\hat{x}}\mathcal{M}[\boldsymbol{v}]
=\mathbb{E}_\omega\!\left[\partial_{\hat{x}}F[\boldsymbol{v}]\right]
=
-\,\mathbb{E}_\omega\!\left[\sum_{\tau\in\mathcal{T}} v_\tau\,\mathbbm{1}\!\left\{L_{\tau\omega}>\hat{x}_{\tau\omega}\right\}\right].
\end{equation}
\end{proposition}

\begin{proof}[Proof]
For each $\tau$, the mapping $\hat{x}_{\tau\omega}\mapsto (L_{\tau\omega}-\hat{x}_{\tau\omega})_+$ is piecewise linear with slope $-1$ when $L_{\tau\omega}>\hat{x}_{\tau\omega}$ and slope $0$ otherwise. Under (i), the kink event occurs with probability zero, so the pathwise derivative exists almost surely. Condition (ii) supplies an integrable dominating bound, allowing differentiation and expectation to be interchanged via dominated convergence, which yields \eqref{eq:ipa-eue}.
\end{proof}

Proposition~\ref{prop:ipa} provides two key interpretations for gradient-based capacity accreditation. First, while calculating Loss of Load Probability (LOLP)-weighted average availability has long been suggested as a simpler heuristic for ELCC, but never formally justified \cite{milligan1997comparison}. We formalize that the ``marginal impact'' of capacity is a well-defined sensitivity of the \emph{expected} shortfall, rather than a numerical artifact arising from a particular perturbation size. 

Second, as illustrated in Fig.~\ref{ipa}, \eqref{eq:ipa-eue} yields an explicit pathwise estimator for this sensitivity: the directional derivative corresponds to the expected \emph{shortage-weighted} availability along a given direction $\boldsymbol{v}$. Unlike finite-difference approaches, which require repeated simulations with different capacity increments and are sensitive to step-size selection, the IPA approach computes gradients directly by aggregating pathwise contributions from a single simulation run. 

In particular, choosing $\boldsymbol{v}=\boldsymbol{A}_g$ yields the gradient direction associated with resource $g$, while choosing $\boldsymbol{v}=\mathbf{1}$ recovers the perfect-resource benchmark defined in Definition~\ref{def:perfect}. As a result, IPA enables gradient-based capacity accreditation for \emph{all} non-storage resources to be obtained from \emph{one} resource adequacy simulation, substantially reducing computational burden while maintaining high estimation accuracy. For storage resources, pathwise sensitivity becomes more complex, as it depends on the full period profile and the assumed dispatch policy.

\subsection{Pathwise Sensitivity}

Proposition~\ref{prop:ipa} implies that, under an EUE-based risk measure, marginal capacity value is determined by \emph{coincidence with scarcity}. Define the baseline scarcity indicator
$S_{\tau\omega}:=\mathbbm{1}\{L_{\tau\omega}>\hat{x}_{\tau\omega}\}$.
Then the directional derivative can be expressed as
\begin{equation}\label{eq:pathwise-inner}
\partial_{\hat{x}}\mathcal{M}[\boldsymbol{v}]
= -\,\mathbb{E}_\omega\!\left[\sum_{\tau\in\mathcal{T}} v_\tau\,S_{\tau\omega}\right].
\end{equation}
Equation~\eqref{eq:pathwise-inner} provides a direct interpretation: an incremental unit of capacity in period $\tau$ reduces EUE only in those scenarios and hours in which the system is already short.

While IPA offers substantial computational convenience for accrediting individual resources, evaluating the \emph{joint} contribution of portfolios that include intertemporal storage resources requires additional care when interpreting \eqref{eq:pathwise-inner}. For conventional generators or renewables, the perturbation direction $\boldsymbol{v}$ can often be treated as exogenous (e.g., $\boldsymbol{v}=\boldsymbol{A}_g$). In contrast, the net injection profile of storage is induced by operational decisions (charging/discharging subject to state-of-charge constraints) and therefore depends on an operating policy and on scenario history. It is thus more appropriate to represent the perturbation direction as $\boldsymbol{v}(\pi)$ for a specified policy $\pi$, potentially scenario-dependent, rather than as a fixed time series \cite{stephen2022impact}. Some optimization-based formulations can be useful for evaluating the MRI of storage resources by explicitly modeling intertemporal operating decisions and state-of-charge constraints \cite{qianstorage}.

This observation also clarifies why the accreditation of a combined resource is generally not additive. When storage is paired with \emph{any} other resources—\emph{whether renewable, gas, or nuclear}— its feasible and optimal actions are coupled to the paired resource’s availability history (through charging opportunities) and to system scarcity conditions. As a result, the marginal contribution of the combined portfolio cannot, in general, be recovered by separately accrediting each component and summing the resulting credits. The interaction arises from the fact that the storage-induced direction $\boldsymbol{v}(\pi)$ is not invariant to the presence or temporal pattern of the paired resource; instead, the portfolio modifies the effective coincidence with scarcity in \eqref{eq:pathwise-inner} through an endogenous reshaping of injections over time.

\begin{figure}[H]
\centering
  \includegraphics[width=0.5\textwidth]{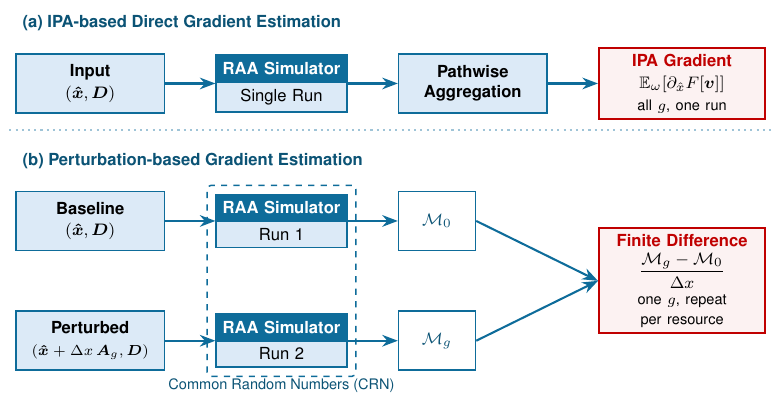}
  \caption{{Comparison between IPA-based direct gradient estimation and perturbation-based gradient estimation, where $\mathcal{M}_{0}=\mathcal{M}(\boldsymbol{\hat{x}},\boldsymbol{D})$ and $\mathcal{M}_{g}=\mathcal{M}(\boldsymbol{\hat{x}}+\Delta x\,\boldsymbol{A}_{g},\boldsymbol{D})$ denote the baseline and perturbed risk-measure outputs.}}
  \label{ipa}
\end{figure}

\section{Perturbation-based Computation and Practical Implications} \label{sec:CC}
While IPA \eqref{eq:ipa-eue} computes gradients directly and efficiently, it requires internal RAA-simulator information that is not always available, and joint portfolios pose further challenges. As Fig.~\ref{ipa} shows, industry implementations instead rely on repeated perturbation-based simulations. This section discusses practical considerations in such simulations and the fundamental computational differences between ELCC- and MRI-based accreditation.

\subsection{Sample Size in Resource Adequacy Assessment} 
Choosing the MC sample size balances computational effort against statistical accuracy. With $Z=\mathcal{M}(\boldsymbol{\hat{x}},\boldsymbol{L})$ and $n$ sampled scenarios $Z_1,\dots,Z_n$, the estimator $\bar{Z}_n=\tfrac{1}{n}\sum_i Z_i$ is, by the CLT, approximately normal with variance $s^2/n$ for large $n$, where $s^2$ is the sample variance.

While increasing $n$ helps achieve a target confidence interval, the nonnegative structure of RA metrics introduces important limitations. Reliability shortfalls are rare events, so most simulated scenarios yield ``all-zero'' outcomes with no unserved energy. Some alternative approaches, such as the \emph{zero-shortfall sample} (CLT+) method and \emph{Bayesian prior}–based estimation, were proposed recently to improve inference accuracy \cite{stephen2024sample}.

\subsection{Step Size in Perturbation-based Gradient Estimation} 

Both ELCC and MRI ultimately require estimating how the risk measure changes under a small perturbation to available capacity. When gradients are approximated by finite differences, one must choose a \emph{perturbation step size} (e.g., a small MW change), which creates a bias--variance tradeoff.

\textbf{Step Size Selection}: Let $\delta > 0$ be a small capacity perturbation. A central difference estimator is given by:
\begin{equation} \label{eq:fd_gradient}
\widehat{\mathrm{MRI}}_g(\delta) = \frac{\mathcal{M}(\boldsymbol{\hat{x}} + \delta \boldsymbol{A}_g, \boldsymbol{L}) - \mathcal{M}(\boldsymbol{\hat{x}} - \delta \boldsymbol{A}_g, \boldsymbol{L})}{2\delta}.
\end{equation}

The choice of $\delta$ involves a fundamental bias–variance trade-off. If $\delta$ is too large, the finite-difference estimator approximates a secant rather than the desired tangent, introducing truncation bias. Conversely, if $\delta$ is too small, the difference in the numerator becomes dominated by Monte Carlo variability, inflating variance through cancellation error. To mitigate this effect, industry implementations typically rely on \emph{Common Random Numbers} (CRN), i.e., using identical random-number streams (or equivalently, the same set of sampled scenarios) for both the baseline and perturbed simulations. This practice induces positive correlation between the paired outputs, which reduces the variance of the difference estimator relative to using independent simulations \cite{glasserman1992crn}. In practice, fixing the number of samples and reusing the same scenarios for baseline and perturbed runs serves a similar purpose by ensuring that differences between runs are attributable primarily to the perturbation rather than sampling variability.

From a statistical perspective, the residual \emph{simulation noise} reflects the uncertainty inherent in Monte Carlo estimates of reliability metrics. In resource adequacy studies, this uncertainty is commonly quantified using the relative standard error (RSE), defined as the ratio of the estimator’s standard deviation to its mean. Components of the estimated change that are on the order of, or smaller than, the RSE are effectively indistinguishable from noise. For example, any part of the result smaller than RSE would be simulation noise in the ISO New England practice.

Theoretically, for central finite-difference gradient estimators under independent simulation noise, minimizing mean-square error leads to an optimal step size that scales with $N^{-1/6}$, reflecting a balance between squared bias $O(\delta^4)$ and variance $O(1/(N\delta^2))$ \cite{zazanis1993convergence}. While CRN and large fixed sample sizes can substantially reduce variance and allow for smaller $\delta$, a practical lower bound remains due to the discrete nature of reliability outcomes. If $\delta$ is smaller than the effective “distance” to the next state change, such as the increment required to trigger an additional shortage hour, then the finite difference $\mathcal{M}(+\delta)-\mathcal{M}(-\delta)$ may evaluate to zero for certain metrics, yielding a vanishing gradient estimate. Consequently, practitioners typically select the smallest perturbation size that produces changes exceeding the RSE threshold, thereby preserving numerical sensitivity while controlling bias.

In industry implementations, different system operators adopt perturbation strategies that reflect both computational considerations and modeling practices. PJM Interconnection computes marginal ELCC using 100~MW increments for both resource-class and unit-specific studies \cite{pjm2025elcc}. ISO New England has proposed smaller perturbations for MRI-based calculations, applying a 0.5~MW perturbation to resource capacities while constructing demand curves by sweeping system capacity in 10~MW increments around the ICR \cite{iso_ne2023_demand_curves}. Through extensive testing, the New York ISO has determined that MRI provides a sufficiently accurate approximation of ELCC at a fraction of the computational cost, and therefore currently applies MRI using a 100~MW perturbation step size in its capacity accreditation framework \cite{nyiso2022icapwg_capacity_accreditation}.

\subsection{Fundamental Difference Between ELCC and MRI-based Capacity Accreditation}

\begin{figure*}[t]
\centering
  \includegraphics[width=1\textwidth]{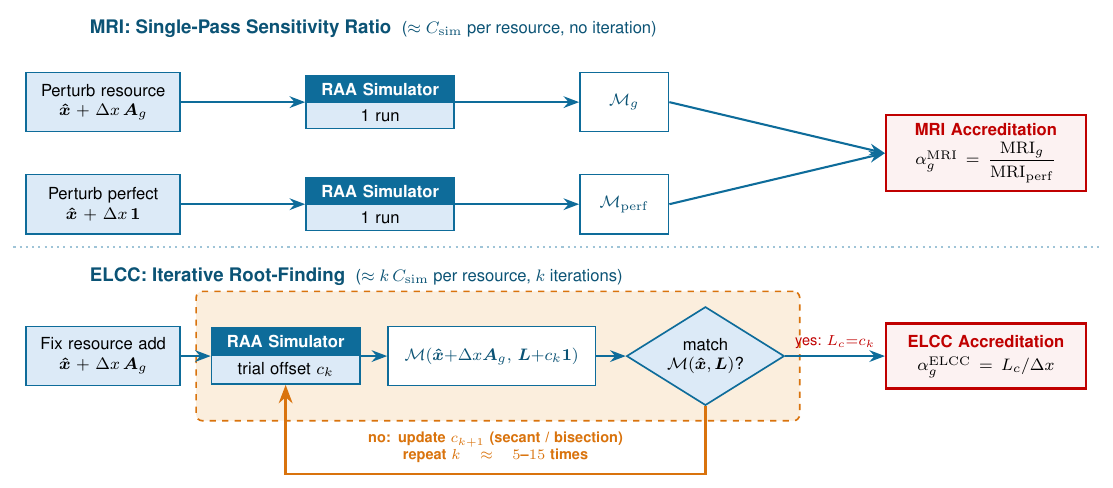}
   \caption{{Computation Process of \emph{Perturbation-based} ELCC and MRI Capacity Accreditation}}
  \label{compare}
\end{figure*}

The MRI ratio is attractive computationally because it avoids solving the implicit equation for $L_c$, but requires only two directional sensitivities, which can be estimated by finite differences within a single simulation run.

Different from the MRI method, the ELCC requires an iterative search procedure to find the equivalent load-carrying capacity. Recall that the ELCC $L_c$ is the root of the nonlinear equation:
\begin{equation}
g(c) = \mathcal{M}\big(\boldsymbol{\hat{x}}+\Delta x\,\boldsymbol{A}_{g},\,\boldsymbol{L}+c\,\mathbf{1}\big) -\mathcal{M}\big(\boldsymbol{\hat{x}},\,\boldsymbol{L}\big).
\end{equation}
Solving this equation for $c = L_c$ creates a significant computational burden characterized by the following factors:

\subsubsection{Iterative Root-Finding}
Unlike the MRI, which is a local linear approximation evaluated at the current system state, ELCC seeks a specific point on the reliability surface where risk is invariant. Since the risk function $\mathcal{M}$ does not have a closed-form inverse, $L_c$ must be found using numerical root-finding algorithms, typically a bisection search method. 

\subsubsection{Nested Simulation Loops}
Fig. \ref{compare} illustrates the computation process of perturbation-based ELCC and MRI capacity accreditation. The computational cost of ELCC is multiplicative. Let $C_{\text{sim}}$ be the cost of one resource adequacy simulation (e.g., a full sequential Monte Carlo year). Assuming the baseline system and perfect increment capacity scenarios have been simulated, then for each study resource, we have
\begin{itemize}
    \item \textbf{MRI Complexity:} The cost is $\approx C_{\text{sim}}$. All marginal accreditation factors can be computed in a single run.
    \item \textbf{ELCC Complexity:} If the root-finding algorithm requires $k$ iterations to converge to a specified tolerance (e.g., $\epsilon_{MW}$), the cost for a \emph{single} resource is $k \times C_{\text{sim}}$.
\end{itemize}
In practice, $k$ often ranges from 5 to 15 iterations depending on the convexity of the risk curve and the required precision. 

\subsubsection{Accelerated Root-Finding via Gradient Methods}

Standard tools typically solve $g(c)=0$ by \emph{bisection}, which is robust to non-smoothness but converges only linearly. Gradient-based root-finding can accelerate this. \emph{Newton's method}, $c_{k+1} = c_k - g(c_k)/g'(c_k)$, uses the sensitivity $g'(c_k)$ of reliability to a load shift, but $g'(c_k)$ is often expensive or unavailable (it requires a full MRI evaluation). We therefore use the \emph{secant method}, a quasi-Newton scheme that keeps the acceleration while avoiding explicit gradients by approximating the derivative from the two previous evaluations $c_k,c_{k-1}$:
\begin{equation}
    c_{k+1} = c_k - g(c_k) \frac{c_k - c_{k-1}}{g(c_k) - g(c_{k-1})}.
\end{equation}
This method offers a superlinear rate of convergence and is significantly faster than bisection. Crucially, it requires only \textbf{one} reliability simulation per iteration (to evaluate $g(c_k)$), as the previous value $g(c_{k-1})$ is cached.

To improve robustness against stochastic noise inherent in Monte Carlo simulations, hybrid root-finding approaches (e.g., \emph{Brent’s method}), which combine bisection with secant-type steps for faster convergence, can also be effective. A detailed discussion of these methods is omitted due to space limitations.

{\subsection{Practical Insights and Alternative Approaches}

\subsubsection{Integration with Existing Accreditation Processes}
The perturbation-based MRI workflow in Fig.~\ref{compare} aligns well with current practice: it requires only a baseline RAA simulation, a perturbed re-run, and a metric difference computed under common random numbers. These capabilities are already available in commercial tools such as GE-MARS and are used in capacity accreditation processes at NYISO and ISO New England~\cite{iso_ne2023_demand_curves,nyiso2022icapwg_capacity_accreditation}. The secant-accelerated ELCC search is compatible with ISOs using ELCC rules, such as PJM's marginal ELCC framework~\cite{pjm2025elcc}. By contrast, the IPA estimator requires access to simulator internals, specifically joint records of per-period shortfall indicators and resource availability states, and is not \emph{directly} applicable to storage resources. Thus, perturbation-based MRI currently captures most of the practical benefit, while future advances in IPA-based methods could further broaden its applicability.

\subsubsection{Comparison with Surrogate, Learning-Based, and Analytical Methods}

\emph{Surrogate and learning-based approaches} replace costly Monte Carlo evaluations with trained approximations~\cite{sharifnia2022mlmc,zhang2025mlmc}. These methods accelerate each evaluation of $\mathcal{M}$, whereas our framework reduces the number of required evaluations. The two approaches are therefore complementary: a surrogate model can be embedded within the MRI or secant-ELCC loop. A key caution, however, is that capacity accreditation depends on \emph{differences} in reliability metrics. A surrogate that is accurate to within a few percent of EUE may still misrank resources whose marginal impacts differ by less than that error level, indicating the need for gradient-aware or ranking-aware training.

\emph{Analytical convolution} method ~\cite{billinton1996reliability} can provide exact sensitivities in stylized reliability models, but it become inefficient for large-scale simulations. In the sequential Monte Carlo setting, IPA plays a role analogous to analytical differentiation in convolution-based methods: it provides exact pathwise sensitivities without finite perturbations, provided the relevant sample-path information is available.}

\section{Numerical Case Study} \label{sec:nc}
\subsection{System Description}
The RTS-GMLC reliability test system represents a contemporary generation portfolio with resource characteristics calibrated to real-world data \cite{gmlc}. The system consists of a single region with 153 generators and a total installed capacity of 14{,}138~MW. The generation mix includes 73 thermal units (coal, gas, oil, and nuclear) totaling 8{,}126~MW, 20 hydro units totaling 1{,}000~MW, 56 solar units totaling 2{,}504~MW, and 4 wind farms totaling 2{,}508~MW. In addition, the system includes one battery storage resource with a power capacity of 50~MW and an energy duration of 3 hours. The system load profile is uniformly scaled to establish a baseline system reliability level close to the `1-in-10' standard. Under this baseline, the system peak load is 9{,}502.7~MW, corresponding to a system-wide LOLD of 2.10~hours/year and EUE of 394.2 MWh/year.

This study uses the Julia-based open-source simulator PRAS \cite{stephen2024probabilistic} with an hourly temporal resolution and a resource capacity resolution of 0.1~MW. {This capacity resolution also sets the ELCC root-finding tolerance: the secant (and bisection) search locates the equivalent firm capacity to the nearest $0.1$~MW, so each ELCC factor resolves to $0.1~\text{MW}/\Delta x$ (e.g., $0.01$ at $\Delta x{=}10$~MW), and ELCC factors are accordingly reported to \emph{two decimals}. The resolution is selected to be fine enough to capture the marginal MRI, whereas using finer steps would further increase the computational burden of the ELCC computations.} For non-intermittent resources, scenarios are sampled from a distribution function representing the outage rate, while for intermittent resources, scenarios are sampled from historical data. To ensure high statistical accuracy, we employ 1{,}000{,}000 Monte Carlo samples for each RA assessment. {The additional sensitivity studies in this revision use $10^5$ samples unless noted, keeping the baseline EUE RSE near 1\% while permitting the larger number of system variants.} All simulations are executed on the Harvard Research Computing cluster: each RA assessment utilizes 2 Intel Xeon Sapphire Rapids CPUs (112 cores total) with 256 GB of memory. The inter-regional transmission constraints are not the focus of this paper and are therefore excluded from the analysis.

\subsection{Capacity Accreditation Factors and Computation Time}

We first select 9 representative generators for gradient-based capacity accreditation, including gas combined-cycle (CC), gas combustion turbine (CT), coal, nuclear, oil, hydro, wind, PV, and battery storage. The unit-level results facilitate a direct methodological comparison. In practice, capacity accreditation is often assigned at the resource-class level based on a weighted average of individual generator accreditation.

For the three perturbation-based methods, i.e. ELCC with bisection, ELCC with secant search, and perturbation-based MRI, we apply a 10~MW perturbation relative to each resource’s nameplate capacity. For renewable resources, the perturbation is applied proportionally to the output profile, while for storage resources, both power and energy capacities are scaled proportionally to preserve duration. In addition, we apply the IPA-based approach to directly compute MRI for non-storage resources, which estimates gradients without requiring repeated perturbation simulations.

\begin{table}[H]
\caption{Comparison of Capacity Accreditation Factors}
\label{tab:comp_credit}
\centering
\begin{tabular}{lccccc}
\toprule
\textbf{Gen ID} & \begin{tabular}[c]{@{}c@{}}\textbf{ELCC} \\ \textbf{(Bisection)}\end{tabular} & \begin{tabular}[c]{@{}c@{}}\textbf{ELCC} \\ \textbf{(Secant)}\end{tabular}  & \begin{tabular}[c]{@{}c@{}}\textbf{MRI} \\ \textbf{(Perturb)}\end{tabular} & \begin{tabular}[c]{@{}c@{}}\textbf{MRI} \\ \textbf{(IPA)}\end{tabular} & \textbf{FOR}\\
\midrule
107 Gas CC   & 0.84 & 0.84 & 0.84 & 0.84 & 3.3\% \\
113 Gas CT   & 0.96 & 0.96 & 0.96 & 0.96 & 3.1\% \\
115 Coal     & 0.90 & 0.91 & 0.91 & 0.91 & 4.0\% \\
121 Nuclear  & 0.47 & 0.48 & 0.48 & 0.49 & 12.0\% \\
101 Oil CT*  & 0.90 & 0.89 & 0.89 & 0.89 & 10.0\% \\
322 Hydro    & 0.76 & 0.75 & 0.76 & 0.75 & -- \\
122 Wind     & 0.09 & 0.08 & 0.09 & 0.08 & -- \\
215 PV       & 0.09 & 0.09 & 0.09 & 0.09 & -- \\
313 Storage  & 0.78 & 0.80 & 0.81 & -- & -- \\
\bottomrule
\multicolumn{6}{r}{\scriptsize *small unit size (20MW)}
\end{tabular}
\end{table}

As shown in Table~\ref{tab:comp_credit}, despite differences in numerical implementation, the resulting capacity accreditation factors are highly consistent across all methods. These results empirically validate the theoretical equivalence between ELCC and MRI and confirm the correctness of the IPA-based direct gradient computation approach.

Table~\ref{tab:comp_time} compares the computational performance of the 4 accreditation methods. To quantify the computational burden, we estimate the total number of RAA runs required for the full system as $ 1 + N_g N_p$. This metric accounts for a single baseline system simulation (the ``1'') plus the individual evaluations required for each of the $N_g$ generators, scaled by the average number of perturbations per resource ($N_p$).

\begin{table}[H]
\centering
\caption{Average computation time per generator (seconds), average perturbation evaluations per resource ($N_p$), and implied total number of RAA runs for whole system.}
\label{tab:comp_time}

\begin{threeparttable}
    \small
    \setlength{\tabcolsep}{4pt}
    \renewcommand{\arraystretch}{1.05}
    
    \begin{tabular}{cccc}
    \toprule
    \multirow{2}{*}{\centering Method}
    & Avg. Time 
    & Avg. \# of  
    & Total \# of \\
    & (per generator) 
    & Perturb. ($N_p$)
    & RAA \\
    \midrule
    ELCC (Bisection) & 452.8 & $\approx$ 7 & $ 1 + N_g N_p$ \\
    ELCC (Secant)    & 307.5 & $\approx$ 3.7 & $ 1 + N_g N_p$ \\
    MRI (Perturb)    & 65.7  & $1$         & $ 1 + N_g N_p$ \\
    MRI (IPA)        & 0.4\tnote{*} & $0$ & $1$ \\
    \bottomrule
    \end{tabular}
    \begin{tablenotes}
        \footnotesize
        \item[*] Calculated by dividing the time of a single IPA run by $N_g$, since one run provides gradients for \emph{all} 153 generator in RTS-GMLC system.
    \end{tablenotes}    
\end{threeparttable}
\end{table}

The MRI-based approach is substantially faster than ELCC. Among perturbation-based methods, MRI achieves an order-of-magnitude speedup relative to bisection-based ELCC. This efficiency arises because MRI requires only one perturbation per resource ($N_p = 1$) to estimate capacity factors from marginal sensitivities, whereas ELCC requires multiple iterative simulations ($N_p \approx 7$) to solve an implicit reliability-equivalence condition. {This count is itself tied to the capacity resolution: the root-finding iterates until the equivalent-capacity bracket narrows to the $0.1$~MW tolerance, so $N_p$ rises as the resolution is refined and falls as it is coarsened.} Within the ELCC framework, the secant method proposed in this paper improves upon standard bisection. By exploiting gradient information to accelerate the root-finding process, the secant method reduces the required perturbations to $N_p \approx 3.7$, yielding a computation time reduction of approximately 32\%.

Finally, the IPA-based gradient estimation method provides the most efficient solution. It requires only a single simulation run (Total \# of RAA $=1$) to compute MRI for \emph{all} non-storage resources simultaneously, effectively decoupling computational cost from the system size $N_g$. For the RTS-GMLC system with 153 generators, we report the effective per-resource time by dividing the single IPA runtime by $N_g$. As shown in Table~\ref{tab:comp_time}, this results in an approximately \textbf{1000$\times$} speedup relative to the conventional bisection-based ELCC method.

{These gains generalize because the relative advantage of MRI over ELCC is governed by the perturbation count $N_p$, a property of the \emph{algorithms} not the system: MRI uses one perturbed assessment per resource, while root-finding ELCC needs $k\geq 3$ regardless of system size, topology, or metric. The per-simulation cost grows with complexity but multiplies all methods equally, so larger systems preserve the relative speedup and amplify the absolute savings. Only the IPA speedup is system-dependent, scaling with $N_g$.}

{\subsection{Monte Carlo Convergence} \label{sec:convergence}

To show the convergence behavior of the underlying Monte Carlo simulations, we repeat the baseline assessment and the perturbation-based MRI accreditation of three representative generators across sample sizes spanning two orders of magnitude. Table~\ref{tab:convergence} reports the baseline EUE estimate with its relative standard error (RSE), together with the resulting accreditation factors, for $n$ ranging from $10^4$ to $10^6$ samples.

\begin{table}[t]

\caption{Monte Carlo Convergence of Baseline EUE and MRI-Based Accreditation Factors (10~MW Perturbation)}
\label{tab:convergence}
\centering
\small
\setlength{\tabcolsep}{4pt}
\begin{tabular}{cccccc}
\toprule
\textbf{Samples} & \begin{tabular}[c]{@{}c@{}}\textbf{EUE} \\ \textbf{(MWh/y)}\end{tabular} & \textbf{RSE} & \begin{tabular}[c]{@{}c@{}}\textbf{Gas CT} \\ $\alpha^{\mathrm{MRI}}$\end{tabular} & \begin{tabular}[c]{@{}c@{}}\textbf{Nuclear} \\ $\alpha^{\mathrm{MRI}}$\end{tabular} & \begin{tabular}[c]{@{}c@{}}\textbf{PV} \\ $\alpha^{\mathrm{MRI}}$\end{tabular} \\
\midrule
$10^4$            & 392.4 & 2.49\% & 0.958 & 0.493 & 0.092 \\
$5{\times}10^4$   & 396.3 & 1.10\% & 0.957 & 0.482 & 0.093 \\
$10^5$            & 395.8 & 0.78\% & 0.957 & 0.485 & 0.092 \\
$5{\times}10^5$   & 393.8 & 0.35\% & 0.960 & 0.489 & 0.092 \\
$10^6$            & 394.6 & 0.25\% & 0.960 & 0.489 & 0.092 \\
\bottomrule
\end{tabular}
\end{table}

The baseline EUE exhibits the canonical $O(n^{-1/2})$ contraction, while the ratio-form factors converge faster: under common random numbers the shared baseline noise cancels between numerator and denominator, so factors at moderate sample sizes already match the $10^6$-sample reference within a few percentage points.
}

\subsection{Perturbation Step Size}

The perturbation step size ($\Delta x$) is critical for both ELCC and MRI. Figure~\ref{perturbation} compares accreditation factors from the 3 methods over $\Delta x$ from 1 to 30~MW for 3 representative generators: a 55~MW Gas CT, a 400~MW Nuclear, and a 115~MW-peak PV plant.

MRI and secant-based ELCC are more robust to $\Delta x$, staying stable over a wide range, whereas bisection-based ELCC is more sensitive, especially for small perturbations, reflecting its repeated evaluation of noisy reliability estimates. This sensitivity is compounded by the simulator's 0.1~MW capacity resolution: perturbations below a few MW may not change system states, amplifying finite-difference noise.

\begin{figure*}[t]
\centering
  \includegraphics[width=0.98\textwidth]{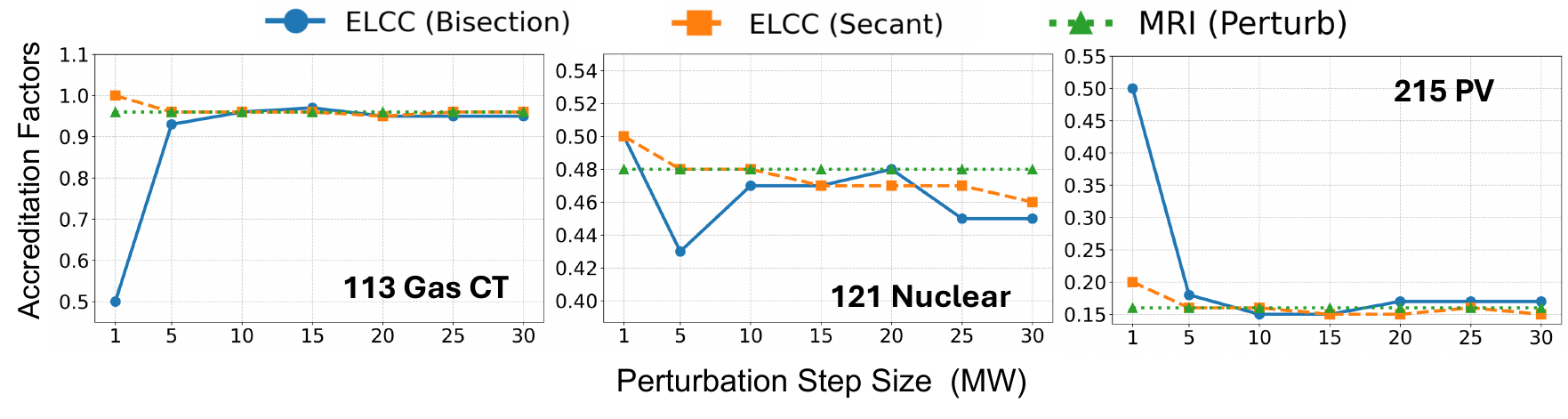}
  \caption{Effect of perturbation step size on accreditation factors for 3 generators}
  \label{perturbation}
\end{figure*}

\begin{table}[H]

\caption{Perturbation Count $N_p$ per Resource Across Perturbation Step Sizes}
\label{tab:stepsize_timing}
\centering
\begin{tabular}{cccc}
\toprule
\textbf{$\Delta x$ (MW)} & \textbf{ELCC (Bisection)} & \textbf{ELCC (Secant)} & \textbf{MRI (Perturb)} \\
\midrule
1  & 3 & 3.7 & 1 \\
5  & 6 & 4.3 & 1 \\
10 & 7 & 4.0 & 1 \\
30 & 8 & 5.0 & 1 \\
\bottomrule
\end{tabular}
\end{table}

{Since the robustness of MRI to step-size selection is a central finding, we also verify that the \emph{computational} advantage reported in Table~\ref{tab:comp_time} is not specific to the 10~MW setting. Table~\ref{tab:stepsize_timing} reports the average number of perturbed assessments per resource across different perturbation steps, where the MRI advantage holds at every step size.

{\subsection{Accreditation under Counting-Based Metrics} \label{sec:lole-results}

While EUE underlies the accreditation in the preceding tables, many system operators retain counting-based reliability standards. Table~\ref{tab:lole} compares two such metrics on the same baseline system: loss-of-load hours (LOLH $\approx$ 2.1~h/yr) and loss-of-load days (LOLD $\approx$ 0.9~day/yr), under both perturbation-based MRI and secant-based ELCC for representative generators.

\begin{table}[H]

\caption{Accreditation Factors under LOLH- and LOLD-Based Metrics (10~MW Perturbation)}
\label{tab:lole}
\centering
\small
\setlength{\tabcolsep}{4pt}
\begin{tabular}{lcccc}
\toprule
\multirow{2}{*}{\textbf{Gen ID}} & \multicolumn{2}{c}{\textbf{LOLH-based}} & \multicolumn{2}{c}{\textbf{LOLD-based}} \\
 & MRI & ELCC & MRI & ELCC \\
\midrule
107 Gas CC   & 0.87 & 0.82 & 0.88 & 0.82 \\
113 Gas CT   & 0.96 & 0.91 & 0.96 & 0.96 \\
115 Coal     & 0.92 & 0.81 & 0.93 & 0.93 \\
121 Nuclear  & 0.53 & 0.50 & 0.55 & 0.54 \\
322 Hydro    & 0.70 & 0.69 & 0.71 & 0.69 \\
122 Wind     & 0.07 & 0.04 & 0.08 & 0.05 \\
215 PV       & 0.08 & 0.02 & 0.03 & 0.03 \\
\bottomrule
\end{tabular}
\end{table}

Both counting metrics yield consistent perturbation-MRI factors that preserve the firm-capacity ordering, and the secant-ELCC estimates agree with MRI to within about $0.05$ for most resources. This agreement is coarser than under EUE because counting metrics change only in discrete event-count increments near the Monte Carlo sampling floor.  LOLD, which aggregates shortfalls at the daily level, is the coarser metric in nominal frequency.}

\subsection{Perturbation Direction}

In addition to perturbation step size, gradient-based accreditation depends critically on the \emph{perturbation direction}, i.e., how incremental capacity is represented within the resource adequacy model. For conventional generators and variable renewables, the perturbation direction can typically be treated as exogenous and proportional to an availability trajectory, so the marginal impact is governed by the coincidence between the added availability and scarcity periods. For storage resources, however, the effective net-injection profile is endogenously determined by intertemporal operating decisions and state-of-charge constraints. This motivates representing storage perturbations as policy-induced directions $\boldsymbol{v}(\pi)$ rather than fixed time series, and implies that joint portfolios involving storage need not exhibit additive marginal impacts.

\begin{table}[H]
\caption{Comparison of EUE and Capacity Accreditation Factors under Different Perturbation Directions}
\label{tab:comp_direction}
\centering
\resizebox{\columnwidth}{!}{%
\begin{tabular}{ccccc}
\toprule
\textbf{Perturbation Direction} & \begin{tabular}[c]{@{}c@{}}\textbf{EUE} \\ \textbf{(MWh/year)}\end{tabular} & \begin{tabular}[c]{@{}c@{}}\textbf{$\Delta$ EUE} \\ \textbf{(MWh/year)}\end{tabular} & \begin{tabular}[c]{@{}c@{}}\textbf{MRI} \\ \textbf{(Perturb)}\end{tabular} \\
\midrule
5 MW Nuclear & 389.00 & 5.22 & 0.48 \\
5 MW PV & 392.51 & 1.72 & 0.16 \\
5 MW Storage & 385.51 & 8.71 & 0.81 \\
5 MW Nuclear, 5 MW Storage & 380.43 & 13.80 & 0.65 \\
5 MW PV, 5 MW Storage & 383.82 & 10.40 & 0.49 \\
\bottomrule
\end{tabular}%
}
\end{table}

Table~\ref{tab:comp_direction} illustrates the role of perturbation direction using perturbation-based MRI under several incremental additions. The second column reports the resulting EUE under each perturbation, while the third column reports the corresponding reduction relative to the baseline system. The joint perturbations are not equal to the sums of the corresponding standalone effects, indicating that marginal impacts depend on the temporal interaction among capacity additions rather than on individual resource contributions alone.

\begin{figure*}[b]
\centering
  \includegraphics[width=0.98\textwidth]{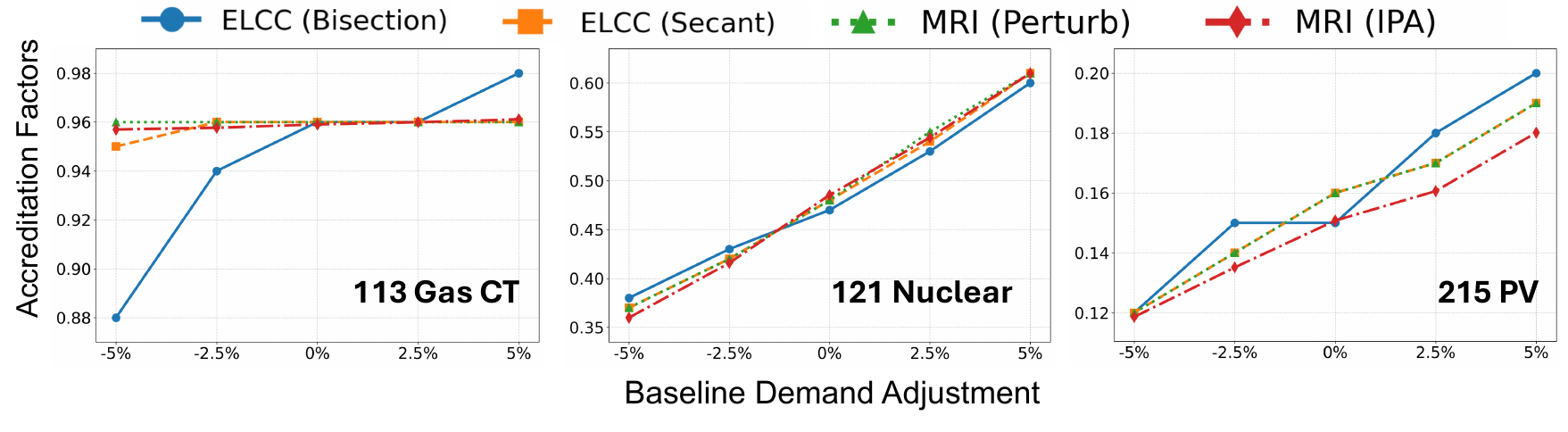}
   \caption{Effect of baseline load adjustment on accreditation factors for 3 generators}
  \label{profile}
\end{figure*}

Two storage-specific limitations follow. First, the IPA estimator in \eqref{eq:ipa-eue} is not directly applicable to storage because the dispatch policy's pathwise derivative generally lacks a closed form; this explains the omitted IPA entry in Table~\ref{tab:comp_credit}. Accordingly, storage credits are estimated here using perturbation-based calculations under the simulator's embedded dispatch policy. Second, these credits are policy-conditional: a different dispatch rule induces a different $\boldsymbol{v}(\pi)$ and hence a different credit. Storage accreditation factors should therefore be interpreted as policy-conditional rather than technology-intrinsic~\cite{qianstorage,stephen2022impact}.

This behavior is consistent with the pathwise interpretation in Section~\ref{sec:IPA}. Storage dispatch shifts energy across hours, thereby modifying the coincidence between incremental capacity and scarcity. Because charging and discharging incentives depend on other resources' availability and on system stress, the direction $\boldsymbol{v}(\pi)$ is scenario-dependent and changes when another resource is added. This non-additivity is not specific to renewable resources, as storage paired with \emph{any} resource can exhibit synergistic or antagonistic interactions. Defensible accreditation of storage-inclusive portfolios therefore requires specifying the operating policy and evaluating joint portfolios directly, rather than summing component credits.}

\subsection{Baseline System Reliability Levels}

We assess sensitivity to system stress by scaling baseline load from $-5\%$ to $+5\%$. Figure~\ref{profile} reports results for the same three generators. Accreditation factors rise monotonically with stress: toward $+5\%$, shrinking reserve margins make each incremental MW more valuable, while under surplus ($-5\%$) factors flatten and converge as the marginal value of capacity diminishes.

Notably, the large nuclear unit keeps a low factor despite its $12\%$ outage rate, a ``lumpiness'' penalty: large-unit outages create deficits precisely during high-stress hours. The gas CT stays high and stable (moderate size, low outage rate), while the PV plant is capped by the coincidence of solar output with net peak load.

{\subsection{Sensitivity to Resource Mix} \label{sec:mix}

Scarcity-coincidence patterns, and therefore accreditation outcomes, depend on the underlying resource mix. To examine this dependence, we construct high-renewable variants of the test system in which wind and solar capacity is scaled to $2\times$ and $3\times$ the reference level, with baseline load recalibrated so that each variant returns to the reference reliability level (matching baseline EUE). This isolates the effect of the mix itself from the effect of overall system stress. Table~\ref{tab:mix} reports MRI-based accreditation factors for representative generators under each mix.

\begin{table}[t]

\caption{MRI-Based Accreditation Factors under Increasing Renewable Penetration (Load Recalibrated to Reference EUE)}
\label{tab:mix}
\centering
\small
\begin{tabular}{lccc}
\toprule
\textbf{Gen ID} & \textbf{Baseline} & \textbf{2$\times$ VRE} & \textbf{3$\times$ VRE} \\
\midrule
107 Gas CC   & 0.846 & 0.860 & 0.865 \\
113 Gas CT   & 0.960 & 0.961 & 0.962 \\
115 Coal     & 0.917 & 0.927 & 0.928 \\
121 Nuclear  & 0.489 & 0.526 & 0.545 \\
322 Hydro    & 0.755 & 0.771 & 0.777 \\
122 Wind     & 0.080 & 0.038 & 0.027 \\
215 PV       & 0.093 & 0.008 & 0.003 \\
\bottomrule
\end{tabular}
\end{table}

Consistent with the pathwise interpretation of Section~\ref{sec:IPA}, as renewable penetration grows residual scarcity concentrates in low-output hours (evening net-load peaks), depressing the marginal credit of the expanding renewable classes while raising that of resources coincident with the shifted scarcity \cite{pfeifenberger2025future}. MRI and ELCC remain equivalent under every mix and the rankings of Table~\ref{tab:comp_time} are unchanged, so the conclusions are not artifacts of the reference mix.
}

{\subsection{Validation on a Simplified ISO New England System} \label{sec:isone}

To verify that the methodological conclusions are not specific to RTS-GMLC, we repeat the accreditation comparison on an independent test system calibrated to the ISO New England footprint used in~\cite{fengmri}. The system uses a 2025 class-aggregated resource mix, has a peak load of approximately 30.8~GW, and is calibrated to a baseline reliability level near the LOLH-based criterion used in ISO-NE planning studies. Table~\ref{tab:isone} reports MRI- and secant-ELCC-based accreditation factors for the largest unit of each resource type under a 10~MW perturbation. The MRI- and ELCC-based factors again agree closely on this system. The computational rankings also remain unchanged, supporting the structural argument that the reported gains reflect properties of the algorithms rather than artifacts of the test system.

\begin{table}[H]

\caption{Accreditation Factors on the Simplified ISO New England Test System (10~MW Perturbation)}
\label{tab:isone}
\centering
\small
\setlength{\tabcolsep}{4pt}
\begin{tabular}{lccc}
\toprule
\textbf{Representative Unit} & \begin{tabular}[c]{@{}c@{}}\textbf{Nameplate}\\\textbf{(MW)}\end{tabular} & \begin{tabular}[c]{@{}c@{}}\textbf{MRI}\\\textbf{(Perturb)}\end{tabular} & \begin{tabular}[c]{@{}c@{}}\textbf{ELCC}\\\textbf{(Secant)}\end{tabular} \\
\midrule
Thermal (TH3)           & 1000 & 0.93 & 0.93 \\
Hydro (IPR7)            & 389  & 0.44 & 0.43 \\
Land-Based Wind (IPR1)  & 564  & 0.24 & 0.23 \\
Offshore Wind (IPR2)    & 480  & 0.42 & 0.41 \\
Solar PV (IPR3)         & 1220 & 0.16 & 0.15 \\
BTM PV (DS-1)           & 1593 & 0.22 & 0.21 \\
\bottomrule
\end{tabular}
\end{table}

}

\section{Conclusion}

This paper develops and applies a gradient-based framework to compare ELCC- and MRI-based capacity accreditation from both theoretical and computational perspectives. Under an EUE-based metric, we establish a local equivalence between marginal ELCC and MRI, while showing that their numerical implementations exhibit fundamentally different computational burdens. In particular, MRI computes accreditation directly from marginal sensitivities within a single Monte Carlo simulation, whereas ELCC relies on iterative root-finding. Among ELCC implementations, secant-based search substantially reduces runtime relative to bisection by exploiting gradient information. We further show that IPA provides a principled and highly efficient mechanism for computing gradient-based accreditation, enabling substantial computational savings when pathwise derivatives are available.

Large-scale numerical results compare the sensitivity of accreditation outcomes to perturbation step size and baseline system reliability level{, as well as to  different reliability metrics and the underlying resource mix, with the main computational conclusions further validated on an simplified test system calibrated to the ISO New England footprint}. Overall, MRI is a practical and scalable alternative to ELCC for large-scale accreditation studies, while still needing careful baseline system specification in gradient-based accreditation workflows. Future research will extend IPA-based accreditation to intertemporal resources such as energy storage and integrate the proposed computational framework into capacity market clearing processes.

\section*{Disclaimer}
The views expressed in this paper are solely those of the authors and do not necessarily represent those of ISO New England, PJM Interconnection L.L.C. or its Board of Managers.

\bibliographystyle{IEEEtran}
\bibliography{ref}
\end{document}